\documentclass[11pt]{article}
\usepackage[margin=1in]{geometry}
\usepackage{latexsym, amscd, amsfonts, mathrsfs, amsmath, amssymb, amsthm, mathtools}
\usepackage{bm}
\usepackage{color}
\usepackage[noend]{algpseudocode}
\usepackage[colorlinks,citecolor=blue,linkcolor=blue,urlcolor=red,pagebackref]{hyperref}
\usepackage[capitalise,nameinlink]{cleveref}

\usepackage{stmaryrd, tikz-cd, enumitem, mathrsfs, bbm, algorithm, url, esint, listings, moreverb, makecell, thm-restate, soul, natbib}

\def\bE{\mathbb{E}}

\def\bP{\mathbb{P}}

\def\bR{\mathbb{R}}
\def\bZ{\mathbb{Z}}

\def\cA{\mathcal{A}}

\def\cE{\mathcal{E}}

\def\cX{\mathcal{X}}

\DeclareMathOperator\Bin{\mathrm{Bin}}

\DeclareMathOperator\supp{\mathrm{supp}}

\DeclareMathOperator\I{\mathsf{I}}

\DeclareMathOperator\s{\mathsf{s}}
\DeclareMathOperator\R{\mathsf{R}}
\DeclareMathOperator\N{\mathsf{N}}
\DeclareMathOperator\avgwdeg{\overline{\wdeg}}
\DeclareMathOperator\degen{\mathsf{degen}}
\DeclareMathOperator\wdeg{\mathrm{wdeg}}

\renewcommand{\setminus}{\backslash}

\newif\ifarxiv
\arxivtrue      %

\ifarxiv
  \newenvironment{proofsketch}
    {\begin{proof}[Proof sketch]}
    {\end{proof}}
\else
  \newenvironment{proofsketch}
    {\begin{proof}\textbf{Sketch} }
    {\end{proof}}
\fi

\newcommand{\rom}[1]{\textup{\uppercase\expandafter{\romannumeral#1}}}

\newtheorem{theorem}{Theorem}
\newtheorem{lemma}[theorem]{Lemma}
\newtheorem{proposition}[theorem]{Proposition}
\newtheorem{corollary}[theorem]{Corollary}

\newtheorem{fact}[theorem]{Fact}

\theoremstyle{definition}

\newcommand{\Parity}[0]{\textrm{PARITY}}

\newcommand{\Threshold}[1]{$#1$\textrm{-THRESHOLD}}
\newcommand{\Conn}[0]{\textrm{Graph-Connectivity}}
\newcommand{\stConn}[0]{$s$-$t$-\textrm{Connectivity}}
\newcommand{\AND}[0]{\textrm{AND}}
\newcommand{\Or}[0]{\textrm{OR}}
\newcommand{\eps}[0]{\varepsilon}
\newcommand{\defeq}[0]{\vcentcolon=}

\begin{document}
\title{A Unified Lower Bound on the Noisy Query Complexity \\ of Boolean Functions}
\author{Yuzhou Gu\thanks{\texttt{sevenkplus.g@gmail.com}} \and Xin Li\thanks{\texttt{lixints@cs.jhu.edu}} \and Yinzhan Xu\thanks{\texttt{xyzhan@ucsd.edu}}}
\date{}

\maketitle

\pagenumbering{gobble}

\begin{abstract}
  We study the query complexity of Boolean functions $f: \{0, 1\}^n \rightarrow \{0, 1\}$ in the noisy query model introduced by Feige, Raghavan, Peleg and Upfal [SICOMP 1994]. In this model, an algorithm can adaptively query the bits of an input vector, but each query result is independently flipped with constant probability $p \in (0, 1/2)$; repeated queries are allowed. The noisy query complexity $\mathsf{N}_p(f)$ of a function $f$ is defined as the minimum expected number of queries needed to compute $f(x)$ with error probability at most $1/3$, for the worst case input $x$.

We prove a general lower bound on $\mathsf{N}_p(f)$ based on degree statistics of certain subgraphs of the Boolean hypercube. This is the first general lower bound beyond those implied by the simple observation that $\mathsf{N}_p(f)$ is lower bounded by the randomized query complexity. We show that this recovers (up to a constant factor) most previously known lower bounds on the noisy query complexity of Boolean functions, providing a unified framework for understanding these results and simplifying the proofs in several cases. Furthermore, this resolves in the affirmative an open problem of Gu, Li and Xu [COLT 2025] that $\mathsf{N}_p(f) = \Omega(\mathsf{I}(f) \log \mathsf{I}(f))$, where $\mathsf{I}(f)$ denotes the total influence of $f$. We also apply our general lower bound to obtain tight bounds on the noisy query complexity for several new functions. 
\end{abstract}

\pagenumbering{arabic}

\section{Introduction} \label{sec:intro}

In modern computation architectures, especially in large systems, errors occur inevitably due to various reasons such as hardware failures, communication errors and environmental noise. It is thus crucial to design algorithms that are robust in the presence of errors. \citet{feige1994computing} formally defined the following extremely natural and basic noisy query model that incorporates errors: Let $f: \{0, 1\}^n \rightarrow \{0, 1\}$ be a Boolean function. Given an unknown fixed input $x \in \{0, 1\}^n$, a (potentially randomized) algorithm can query each bit $x_i$, but the returned bit could be flipped with some fixed constant probability $p$, and repeated queries are allowed. The goal is to compute $f(x)$ while minimizing the number of queries. There has been a lot of work studying this model~\citep{reischuk1991reliable,feige1992complexity,DBLP:journals/rsa/KenyonK94,evans1998average,zhu2023optimal, zhu2023noisy,wang2025noisy,gu2025tight}, as well as the related noisy comparison model \citep{horstein1963sequential,berlekamp1964block,burnashev1974interval,feige1992complexity,ben2008bayesian,Emamjomeh-Zadeh16,DereniowskiTUW19,wang2022noisy,gu2023optimal,wang2023variable,zhu2023optimal,DBLP:conf/stacs/DereniowskiLU25}.

However, despite much attention, even such a basic model has not been fully understood. Let $\N_p(f)$ denote the minimum expected number of noisy queries needed to compute $f$ with error probability $\le 1/3$ by an adaptive randomized algorithm. Clearly, $\N_p(f) \ge \Omega(\R(f))$, the randomized  query complexity for computing $f$ within error $1/3$, as computation with exact queries can only be easier than computation with noisy queries. Also, $\N_p(f) = O(\R(f) \log \R(f))$\footnote{Throughout the paper, we use the convention that $x \log x = 0$ for $x = 0$.} for the following reason: Given a randomized algorithm for computing $f$ with exact queries, one can repeat each query $O(\log \R(f))$ times in the noisy query model and take the majority results, so that each query will be correct with probability $1 - O(1 / \R(f))$. Hence, by union bound, all $\R(f)$ queries can be simulated correctly with constant probability. However, this $[\Omega(\R(f)), O(\R(f) \log \R(f))]$ range is not strong enough to fully determine $\N_p(f)$, and there exist $f$ where $\N_p(f) = \Theta(\R(f))$ and others where $\N_p(f) = \Theta(\R(f) \log \R(f))$. For instance, it is known that \Parity{} \citep{feige1994computing} and \Conn{}~\citep{gu2025tight} have $\N_p(f) = \Theta(\R(f) \log \R(f))$, and \AND{}, \Or{} have $\N_p(f) = \Theta(\R(f))$~\citep{feige1994computing}. Hence, it is not possible to fully characterize $\N_p(f)$ using only $\R(f)$, and it is also important to study $\N_p(f)$ with respect to other complexity measures.

Two important measures of Boolean functions are their sensitivity and influence. For any $x \in \{0, 1\}^n$, the sensitivity of $f$ at $x$, $\s(f, x)$, is defined as the number of $y$ such that $f(x) \ne f(y)$ and $x, y$ only differ at one bit; and the maximum sensitivity of $f$ at any input $x$ is denoted as $\s(f)$. On the other hand, the total influence $\I(f)$ is the same as the average sensitivity over all inputs, i.e., $\I(f) = \bE_{x \sim \{0, 1\}^n}[\s(f, x)]$. Both sensitivity and influence are central notions in Boolean function analysis and have been used in various contexts such as
hardness of approximation \citep{dinur2005hardness}, voting theory \citep{mossel2012quantitative} and random graph theory \citep{friedgut1996every}. Sensitivity and its variants have also been shown to be closely related to many query complexity measures (e.g., \citep{DBLP:journals/cc/NisanS94}).

It is well-known that $\R(f) = \Omega( \s(f))$ \citep{nisan1989crew} and hence we also have $\N_p(f) = \Omega(\s(f))$. In \citep{reischuk1991reliable}, the authors proved that the \emph{non-adaptive} noisy query complexity of any Boolean function $f$ is at least $\Omega(\s(f)\log \s(f))$, so a natural question is whether this can be extended to adaptive queries. This is not possible. For instance, the $\mathrm{OR}$ function has sensitivity $n$ and adaptive noisy query complexity $O(n)$ \citep{feige1994computing}. 

This motivates \citet{gu2025tight} to study the relationship between $\N_p(f)$ and $\I(f)$, as $\I(f) = \bE_{x \sim \{0, 1\}^n}[\s(f, x)]$ is the average sensitivity of $f$ which can be smaller than $\s(f)$. As before, the simple fact that $\N_p(f) \geq \Omega(\R(f))$ also implies $\N_p(f) = \Omega(\I(f))$. Can there be better bounds on $\N_p(f)$ based on $\I(f)$? \citet{gu2025tight} conjectured that $\N_p(f) = \Omega(\I(f) \log \I(f))$, and proved it for the special case where function $f$ has near-maximum $\Theta(n)$ total influence $\I(f)$. In this work, we resolve their conjecture in the affirmative.

\begin{theorem} \label{thm:main}
For every constant $p \in (0, 1/2)$, there exists a constant $c=c(p)>0$ such that $\N_p(f) \ge c \I(f) \log \I(f)$ for every Boolean function $f$.
\end{theorem}

This is the first lower bound on the noisy query complexity that is general enough to apply to an arbitrary function (besides the trivial $\Omega(\R(f))$ lower bound). Previously, lower bounds are often studied for specific functions \citep{feige1994computing},  restricted classes of functions such as random functions~\citep{reischuk1991reliable,feige1992complexity,evans1998average} and functions with high total influence \citep{gu2025tight}, or restricted algorithms such as non-adaptive algorithms \citep{reischuk1991reliable}. This is also the best lower bound of $\N_p(f)$ based purely on $\I(f)$ one could hope for, because for any value of $v \in [n]$, there are functions where $\I(f) = v$ and $\N_p(f) = \Theta(\I(f) \log \I(f))$ (e.g., \Parity{} on a subset of $v$ indices).

\subsection{Main Technical Theorem and Other Applications}

\Cref{thm:main} is a consequence of an even stronger result which we will explain next. Consider the Boolean hypercube $H$. For every edge $(x, y) \in E(H)$, we keep it if $f(x) \ne f(y)$ and remove it otherwise. Let us call the remaining graph $H_f$. Our main technical theorem that leads to \Cref{thm:main} is the following: 

\begin{restatable}{theorem}{MainTHM}
\label{thm:main-general}
  For every constant $p\in (0, 1/2)$ and $\gamma \in (0, 1]$, there exists $c=c(p,\gamma)>0$ such that the following holds.
  Let $f: \{0, 1\}^n \to \{0, 1\}$ be a Boolean function and 
  let $G$ be an induced subgraph of $H_f$ with a nonempty edge set.
  Let $d_1 = \min_{x\in V(G)} \deg_G(x)$. 
  Let $\mu$ be the distribution on $V(G)$ defined as
  $ \mu(x) = \frac{\deg_G(x)}{2 |E(G)|} $
  and let $d_2 \ge d_1$ be any number satisfying $\bP_{x\sim \mu}(\deg_G(x) \ge d_2) \ge \gamma$.
  Then $\N_p(f) \ge c \cdot d_2 \log (1+d_1)$.
\end{restatable}
Note that the distribution $\mu$ is equivalent to the following distribution: We first sample a random edge from $E(G)$, and then sample a random endpoint of the edge. Using this interpretation, our theorem works for the case where $G$ has a constant fraction of edges that have one endpoint of degree at least $d_2$.

To see how \Cref{thm:main-general} implies \Cref{thm:main}, we consider the degeneracy of a Boolean function $f$, which may be a parameter of independent interest (even though degeneracy is a well-studied concept in graph theory, the degeneracy of $H_f$ has not been studied to the best of our knowledge). The degeneracy of a graph is the smallest integer $k$ such that every subgraph of the graph has at least one vertex with degree at most $k$. We use $\degen(f)$ to denote the degeneracy of $H_f$. By definition of $\degen(f)$, there exists an induced subgraph $G$ of $H_f$ where the minimum degree of $G$ is $\degen(f)$. Hence, we can apply \Cref{thm:main-general} with $d_1 = d_2 = \degen(f)$ to obtain the following result:

\begin{theorem} \label{thm:degen}
For every constant $p \in (0, 1/2)$, there exists a constant $c=c(p)>0$ such that $\N_p(f) \ge c \degen(f) \log \degen(f)$ for every Boolean function $f$.
\end{theorem}

Note that the average degree of $H_f$ is exactly $\I(f)$. A basic fact in graph theory is that the average degree of a graph is at most twice its degeneracy (see e.g., \cite{lick1970k}). Applying this basic fact to $H_f$, we get that $\I(f) \le 2\degen(f)$. Therefore, \Cref{thm:degen} implies \Cref{thm:main}.

We also remark that $\degen(f)$ is related to sensitivity. More precisely, it can be seen that $\degen(f)=\max_{\emptyset \subsetneq V \subseteq \{0, 1\}^n}\min_{x \in V}s(f_V, x)$, where $f_V$ is the partial function which agrees with $f$ on $V$ but undefined everywhere else.

We also apply \Cref{thm:main-general} to obtain several new lower bounds for more specific functions. The first example is biased random function, which extends the previous lower bound on unbiased random functions  \citep{reischuk1991reliable,feige1992complexity,evans1998average} to a more general setting (in the following, a biased random function means that for each input $x \in \{0,1\}^n$, the value $f(x)$ is sampled independently, with $\Pr[f(x)=1]=q$, where $q$ need not equal $1/2$): 

\begin{restatable}{theorem}{Biased}
\label{thm:biased}
    Let $\lambda \in (0, 1), p\in (0, 1/2), \gamma > 0$ be fixed constants. 
    Let $f: \{0, 1\}^n \rightarrow \{0, 1\}$ be a biased random function where for each $x \in \{0, 1\}^n$, $f(x) = 1$ independently with probability $q = \gamma n^{-\lambda}$. Then $\N_p(f) \ge \Omega(n \log n)$ with probability $1 - \exp(-c n)$ for some constant $c > 0$ over the randomness of $f$.  
\end{restatable}

Our second example considers the complexity of counting the number of any fixed pattern $t = t_1 t_2 \cdots t_k$ in the bit sequence $x_1 x_2  \cdots x_n$ as a subsequence. This is a vast generalization of the problem on counting the number of $1$'s in the sequence $x_1 x_2  \cdots x_n$, which was known to require $\Omega(n \log n)$ noisy queries \citep{feige1994computing}. Note that this subsequence counting problem does not have Boolean output, but our techniques can still apply to such functions, via reduction from a threshold version of the function which has Boolean output. Also, $\N_p(f)$ for non-Boolean functions $f$ can be defined analogously. 

\begin{restatable}{theorem}{Subseq}
\label{thm:subseq}
    Let $t \in \{0, 1\}^k$ be a fixed pattern for some $k = O(1)$. Let $f: \{0, 1\}^n \rightarrow \bZ$ be a function where $f(x)$ denotes the number of times $t$ appears in $x_1, \ldots, x_n$ as a (not necessarily contiguous) subsequence. Let $p \in (0, 1/2)$ be a constant. Then $\N_p(f) = \Omega(n \log n)$. 
\end{restatable}

Besides showing new lower bounds on $\N_p(f)$, \Cref{thm:main-general} is strong enough to unify \emph{all} previous nontrivial (i.e. $\N_p(f) \gg \R(f)$) asymptotic noisy query lower bounds to the best of our knowledge. Previously, nontrivial lower bounds for $\N_p(f)$ were known for \Parity{}, \Threshold{k}, symmetric functions \citep{feige1994computing}, random functions \citep{reischuk1991reliable,feige1992complexity,evans1998average},  \Conn{} and \stConn{} \citep{gu2025tight}, and functions with $\Omega(n)$ total influence \citep{gu2025tight}. There are also examples whose lower bounds follow by directly reducing from these problems (e.g., $s$-$t$ Shortest Path \citep{gu2025tight}). 
Besides such examples, this is a comprehensive list to the best of our knowledge. Among these problems, functions with $\Omega(n)$ total influence are a special case of \Cref{thm:main}. \Parity{} and random functions have $\Theta(n)$ total influence, so they are also special cases of \Cref{thm:main}. In \Cref{sec:app}, we will explain how \Cref{thm:main-general} implies the lower bounds of noisy query complexity for \Threshold{k}, symmetric functions, \Conn{} and \stConn{}. Because \Cref{thm:main-general} is stronger than previous approaches, the new lower bound proofs based on it are also much simpler than previous proofs in some cases. For instance, to show the lower bound for \Conn{}, \citet{gu2025tight} had to analyze the structure of a uniform spanning tree of a complete graph, which is quite complicated and takes $5$ pages of their proofs. Our proof based on \Cref{thm:main-general} is quite simple and does not rely on \citet{gu2025tight}'s structural results. 

We remark that the lower bound obtained from \Cref{thm:main-general} is not tight in general. The best lower bound one can hope to obtain using \Cref{thm:main-general} is $\Omega(\Delta(H_f) \log \Delta(H_f))$, where $\Delta(H_f)$ denotes the maximum degree of $H_f$. As $\Delta(H_f)$ is exactly the sensitivity of $f$, $\s(f)$, the best lower bound we can hope to obtain is $\Omega(\s(f) \log \s(f))$. Recall $\N_p(f) = \Omega(\R(f))$, so as long as $\R(f) = \omega(\s(f) \log \s(f))$ (see e.g., \citep{rubinstein1995sensitivity} for examples), our lower bound is not tight. There are many known examples with $\R(f) = \omega(\s(f) \log \s(f))$. 

\section{Preliminaries} \label{sec:prelim}
For a simple undirected graph $G$, $V(G)$ denotes the set of vertices and $E(G)$ denotes the set of edges. For $x\in V(G)$, $\deg_G(x)$ denotes the degree of $x$ in $G$.

For $x\in \{0,1\}^n$, $|x|$ denotes the Hamming weight of $x$. For a Boolean function $f: \{0, 1\}^n \rightarrow \{0, 1\}$, $H_f$ denotes the graph whose vertices are $\{0, 1\}^n$ and there is an edge between two vertices $x$ and $y$ if and only if $x$ and $y$ differ at exactly one bit and $f(x) \ne f(y)$. 

$\Bin(n,p)$ denotes the binomial distribution. The Chernoff bound states that for $X\sim \Bin(n,p)$, $\mu=np$, and $\epsilon\ge 0$, we have
\begin{align*}
  \bP(X \ge (1+\epsilon)\mu) &\le \exp\left(-\frac{\epsilon^2}{2+\epsilon} \mu\right), \\
  \bP(X \le (1-\epsilon)\mu) &\le \exp\left(-\frac{\epsilon^2}{2} \mu\right).
\end{align*}
We also use the Azuma-Hoeffding inequality. Let $g: \{0, 1\}^n \rightarrow \bR$ where 
\begin{align*}
\left| g(x_1, \ldots, x_n) - g(y_1, \ldots, y_n)\right| \le B
\end{align*}
for some bound $B$, for every $x$ and $y$ that differ at exactly one bit. Then for every $\lambda > 0$, 
\begin{align*}
    \bP_{x \sim \{0, 1\}^n} \left(\left|g(x) - \bE[g(x)] \right| \ge \lambda B\right) \le 2 \exp(-2 \lambda^2 / n).
\end{align*}

Markov's inequality states that for a non-negative random variable $X$, $\bP(X \ge t \bE[X]) \le 1/t$ for all $t>0$.
We will also use the following corollary:

\begin{corollary}
\label{cor:markov}
Let $X$ be a random variable taking values in $[0, 1]$, then $\bP(X \ge t) \ge \bE[X]-t$ for $0\le t\le \bE[X]$. 
\end{corollary}
\begin{proof}
We have
\begin{align*}
  \bP(X \ge t) = \bP(1-X \le 1-t) \ge 1-\bP(1-X \ge 1-t) \ge 1-\frac{1-\bE[X]}{1-t} \ge \bE[X]-t,
\end{align*}
where the third step is by Markov's inequality.
\end{proof}

\section{Technical Overview} \label{sec:overview}

In this section, we give a high-level overview of our proof for the main technical theorem \Cref{thm:main-general}. \ifdefined\iscolt{The full proof is deferred to \Cref{sec:proof}.} \fi Our proof uses the three-phase framework developed in \citet{feige1992complexity, gu2025tight}. Given an input $x \in \{0, 1\}^n$, the goal of the three-phase problem is also to compute $f(x)$ by querying the input bits $x_i$. In the first phase, the algorithm queries every bit $c_1\log d_1$ times in a nonadaptive way with noisy queries. In the second phase, some information is revealed to the algorithm for free. In the third phase, the algorithm can adaptively make $c_2 d_2$ \emph{exact} queries to compute $f(x)$ (for simplicity, in this overview we assume a worst case $c_2 d_2$ query bound; in the actual proof, we relax it to an expected bound). As shown in \citet{gu2025tight}, if an algorithm cannot solve the three-phase problem with constant error probability, then any algorithm $\cA$ making $c_1 c_2 d_2 \log(d_1)$ noisy queries cannot compute $f(x)$ with constant error probability. Essentially, if there is such an algorithm $\cA$ for computing $f(x)$, then we can simulate it for the three-phase problem. As shown in \citet{gu2025tight}, there is some way to design the second phase, so that the observations made in the first two phases are equivalent to the following:
\begin{itemize}
    \item For every $i$ where $x_i = 1$, $x_i$ is revealed independently with probability $p_+$, where $p_+ = 1 - d_1^{-c_3 \pm o(1)}$ for some constant $c_3$;
    \item For every $i$ where $x_i = 0$, $x_i$ is revealed independently with probability $p_-$, where $p_- = o(1)$. 
\end{itemize}
Let $z \in \{0, 1, \star\}^n$ denote the revealed information up to the second phase, i.e., $z_i = x_i$ if $x_i$ is revealed, and $z_i = \star$ otherwise. Let $\nu_x$ be the distribution of $z$ under $x$.

Previous methods using the three-phase framework \citep{gu2025tight} (or the predecessor two-phase framework in the case of \citet{feige1994computing}) use ad-hoc analyses tailored to the specific functions. We instead provide a unified approach, based on degree statistics of a certain subgraph of the Boolean hypercube. Our approach first differs from \citet{gu2025tight} by defining a distribution $\mu$ that applies to any Boolean function. 
Recall the distribution $\mu$ and a subgraph $G$ of $H_f$ from the statement of \Cref{thm:main-general}, where $\mu(x) = \frac{\deg_G(x)}{2 |E(G)|}$ for $x \in V(G)$. We can extend $\mu$ to $\{0, 1\}^n$, simply by setting $\mu(x) = 0$ for $x \notin V(G)$. This will be the hard distribution for the three-phase problem. 

Our goal is to show that, after the three phases, the probability of $f(x) = 0$ and the probability of $f(x) = 1$ are both likely lower bounded by a constant, where $x$ is distributed with respect to the posterior distribution under all the observations. If this holds, then no matter what value the algorithm outputs for $f(x)$, there is always a constant error probability. Our approach for showing such balance conditions is different from previous approaches \citep{feige1994computing,gu2025tight}. Previously, they were shown either by ad-hoc analyses or, for the case of high-influence functions \citep{gu2025tight}, by the fact that $\I(f) = \Theta(n)$ implies this balance condition. Our approach is more refined and, as we will see shortly, uses a graph-theoretical view of this condition.

As a warm-up, we will first describe why the probability of $f(x) = 0$ and the probability of $f(x) = 1$  are both likely lower bounded by a constant after the first two phases. This is necessary for the lower bound, because otherwise, the algorithm can simply return the value that is more likely, and achieve sub-constant error probability without any queries in the third phase.  

Let $\mu_z$ denote the posterior probability of $x$ under the observation $z$. We need to show that $\min\{\mu_z(f^{-1}(0)), \mu_z(f^{-1}(1))\} \ge \Omega(1)$, with good probability. Let $\cX_z$ denote the support of $\mu_z$. We can imagine a graph, whose vertices are $\cX_z$, and there is an edge between two vertices $x$ and $y$ if $(x, y) \in E(G)$. Each vertex $x$ has a weight, which equals $\mu_z(x)$, and we will also assign each edge $(x, y)$ a weight $w_z(x, y)$. The degree of a vertex $x$ is defined as 
\[
\wdeg_z(x) = \mu_z(x)^{-1} \sum_{y: (x, y) \in E(G)} w_z(x,  y). 
\]
This is a bipartite graph, where one side is $\cX_z \cap f^{-1}(0)$, and the other side is $\cX_z \cap f^{-1}(1)$. If we can show an upper bound $D$ on the maximum degree, and show a lower bound $M$ on the total edge weight, then the total vertex weights of either side of the bipartite graph can be lower bounded by $M / D$, i.e., $\min\{\mu_z(f^{-1}(0)), \mu_z(f^{-1}(1))\} \ge M / D$. The values for $M$ and $D$ we aim for are $M = \Omega(\sqrt{1-p_+})$ and $D = O(\sqrt{1-p_+})$, which can show $\min\{\mu_z(f^{-1}(0)), \mu_z(f^{-1}(1))\} \ge \Omega(1)$ via the previous inequality. More precisely, we will show an upper bound $D$ on the degree of most probability mass of $x$ under $\mu_z$, which is still able to give us a lower bound. 

Towards this goal, we define $w_z(x, y) = \sqrt{\frac{\mu_z(x) \mu_z(y)}{\deg_G(x) \deg_G(y)}}$.
Let $N_+(x)$ denote the set of $y$ where $(x, y) \in E(G)$ and $|y| = |x| + 1$, and let $N_-(x)$ denote the set of $y$ where $(x, y) \in E(G)$ and $|y| = |x| - 1$. By some calculation, one can show the following equality (this is \Cref{eqn:wdeg-nbrs} in \Cref{sec:proof}) that relates the degree of $x$ with the number of vertices in $N_+(x)$ and in $N_-(x)$ that also appears in the support of $\mu_z$. 
\begin{align}
\label{eqn:wdeg-nbrs-overview}
  \wdeg_z(x) = \deg_G(x)^{-1} \left( |N_+(x)\cap \cX_z| \sqrt{\frac{1-p_+}{1-p_-}} + |N_-(x) \cap \cX_z| \sqrt{\frac{1-p_-}{1-p_+}} \right).
\end{align}

\paragraph{Upper bound on degree of ``most'' vertices. } Because $|N_+(x)\cap \cX_z| \le \deg_G(x)$, the contribution of the first term in \Cref{eqn:wdeg-nbrs-overview} is only $\sqrt{\frac{1-p_+}{1-p_-}} = O(\sqrt{1-p_+})$ (recall $p_- = o(1)$). Next, we consider the second term. For $z \sim \nu_x$, for any $y \in N_-(x)$, $y \in \cX_z$ if $z_i = \star$, where $i$ is the index that $x$ and $y$ differ. This happens with probability $1 - p_+$, and is independent for different $y \in N_-(x)$. Thus, we can apply concentration bounds to show that it is likely that
\begin{align*}
|N_-(x) \cap \cX_z| \le O((1-p_+) \cdot |N_-(x)|) = O((1-p_+) \cdot \deg_G(x)). 
\end{align*}
If this holds, then the contribution of the second term in \Cref{eqn:wdeg-nbrs-overview} can also be bounded by $O(\sqrt{1 - p_+})$. This is where we need the condition that the minimum degree of $G$ is $d_1$; otherwise, the deviation probability of the above concentration becomes too large. 

\paragraph{Lower bound on total edge weights. } We will show that with constant probability (over the randomness of $z$), $\wdeg_z(x) = \Omega(\sqrt{1-p_+})$ for a constant density of $x$ under $\mu_z$. If this holds, then the total edge weights can be lower bounded by 
\[
\frac{1}{2} \bE_{x \sim \mu_z}[\wdeg_z(x)] \ge \Omega(\sqrt{1-p_+})
\]
as desired. 

By the assumption in the statement of \Cref{thm:main-general}, for $x \sim \mu$, with constant probability, we have $\deg_G(x) \ge d_2$. Hence, either $|N_+(x)| \ge \deg_G(x) / 2 \ge d_2 / 2$, or $|N_-(x)| \ge \deg_G(x) / 2 \ge d_2 / 2$, with constant probability. We will assume without loss of generality that the former case holds. Fix any $x$ where $|N_+(x)| \ge \deg_G(x) / 2 \ge d_2 / 2$. As before, when $z \sim \nu_x$, some $y \in N_+(x)$ remains in $\cX_z$  when $z_i = \star$, where $i$ is the index of the bit where $x$ and $y$ differ. Again, we can apply concentration bounds to show that it is likely 
\begin{align*}
|N_+(x) \cap \cX_z| \ge \Omega((1-p_-) \cdot |N_+(x)|) = \Omega((1-p_-) \cdot \deg_G(x)). 
\end{align*}
Hence, by \Cref{eqn:wdeg-nbrs-overview}, $\wdeg_z(x) \ge \Omega(\sqrt{1-p_+})$ as desired. 

\paragraph{Handling the third phase. } In the third phase, the algorithm can make $c_2 d_2$ adaptive exact queries. Each query reveals one bit of the input vector $x$, and the posterior distribution of $x$ will change. In particular, the posterior probability of those $x$ that are inconsistent with the revealed bit will become $0$, and the posterior distribution of other $x$'s will be scaled properly. Let $\mu_{z, t}$ denote the posterior distribution of $x$ after the first $t$ queries and let $\cX_{z, t}$ denote the support of $\mu_{z, t}$. 
We can consider the subgraph of the graph on $\cX_z$ considered above, on the subset of vertices $\cX_{z, t}$. All the vertex weights should be scaled to $\mu_{z, t}(x)$, and all edge weights should also be scaled properly. Intuitively, because now we are considering a subgraph, the upper bound on the degrees of ``most'' vertices should still hold. 

To show a lower bound for the total edge weights, let $X$ denote the set of $x$ where $|N_+(x)| \ge \deg_G(x) / 2 \ge d_2 / 2$ and $|N_+(x) \cap \cX_z| = \Omega(\deg_G(x))$. Previous steps can establish that $\mu_z(X) \ge \Omega(1)$ with constant probability. It is not difficult to check that $\bE[\mu_{z, t}(X) \mid z] = \mu_z(X)$, so $\bE[\mu_{z, t}(X) \mid z] \ge \Omega(1)$ with constant probability (over $z$) as well. Hence, with constant probability, vertices in $X$ still carry a constant fraction of the weights in the subgraph on $\cX_{z, t}$. Because each query in the third phase reveals one bit, $|N_+(x) \cap \cX_{z, t}| \ge \Omega(\deg_G(x)) - t$ for any $x \in X$. As $t \le c_2 d_2$, for sufficiently small $c_2$, we still have $|N_+(x) \cap \cX_{z, t}| \ge \Omega(\deg_G(x)) - c_2 d_2 = \Omega(\deg_G(x))$. Hence, we can still use a formula similar to \Cref{eqn:wdeg-nbrs-overview} to obtain a lower bound of total edge weights. 

\section{Proof of the Main Technical Theorem} \label{sec:proof}
In this section, we prove \Cref{thm:main-general}, which we recall below:
\MainTHM*

We will show that there exists a constant $D$ so that theorem statement holds whenever $d_1\ge D$ (this implies $d_2 \ge D$ as well because $d_2 \ge d_1$).
For $d_1\le D$, the result follows from $\N_p(f) = \Omega(\R(f)) = \Omega(\max_x \s(f,x)) = \Omega(d_2)$.
In this proof, asymptotic notations (such as $o(\cdot)$) are with respect to $d_1\to \infty$. In particular, $o(1)$ denotes a function that goes to $0$ as $d_1\to\infty$. By taking $D$ to be sufficiently large, we can assume these $o(1)$ terms are smaller than an arbitrary constant.

\subsection{The Three-Phase Framework}

We use the three-phase framework developed in \citet{feige1994computing,gu2025tight}.
Consider the following three-phase problem.
\begin{itemize}
  \item Phase 1: The algorithm nonadaptively queries each bit $m_1 = c_1 \log d_1$ times using noisy queries for some constant $c_1>0$.
  \item Phase 2: The oracle reveals certain input bits for free.
  \item Phase 3: The algorithm adaptively makes $m_2 = c_2 d_2$ exact queries in expectation for some constant $c_2>0$. 
\end{itemize}
After the three phases, the algorithm needs to output $f(x)$ with success probability at least $1-\delta$ for some constant $\delta>0$.

The following reduction comes from \citet{gu2025tight}.

\begin{proposition} \label{prop:three-phase}
  If there is an algorithm that computes $f(x)$ with success probability at least $1-\delta$ in the worst case using $c_1c_2 d_2 \log d_1$ queries in expectation in the noisy query model, then it can solve the corresponding three-phase problem with success probability at least $1-\delta$ in the worst case.
\end{proposition}
\begin{proofsketch}
Let $\cA$ be an algorithm for computing $f(x)$ in the noisy query model. 
    The three-phase problem can simulate the noisy query model and supply $\cA$ with noisy bits. Because of Phase 1, this simulation can already support the first $m_1 = c_1 \log d_1$ noisy queries to each bit $x_i$. If $\cA$ makes more than $m_1$ queries to some bit $x_i$, the three-phase problem makes an exact query for $x_i$ in Phase 3, and can simulate all future noisy queries to $x_i$. Since $\cA$ makes  $c_1 c_2 d_2 \log d_1$ queries in expectation, the expected number of $x_i$ to which $\cA$ makes more than $m_1$ queries is at most $(c_1 c_2 d_2 \log d_1) / m_1 = c_2 d_2$. Hence, the three-phase problem makes at most $c_2 d_2$ queries in Phase 3 in expectation. 
\end{proofsketch}

In the definition of $\N_p(f)$, we require the algorithm to be correct with probability at least $2/3$. By repeating $O(\log(1/\delta)) = O(1)$ times and taking the majority answer, the success probability can be boosted to $1-\delta$. Therefore, to prove \Cref{thm:main-general}, it suffices to show the following hardness result for the three-phase problem.
\begin{proposition} \label{prop:three-phase-hard}
  There exists $c_1,c_2,\delta>0$ such that no algorithm can solve the corresponding three-phase problem with success probability at least $1-\delta$ in the worst case.
\end{proposition}

Let $G$ be the induced subgraph of $H_f$ as in the statement of \Cref{thm:main-general} and $\mu$ be the corresponding distribution on $V(G)$ where $\mu(x) = \frac{\deg_G(x)}{2|E(G)|}$. 
We will show that the three-phase problem is hard even if the input $x$ is drawn from the distribution $\mu$.

\subsection{Phase 1 and 2}
Let $m_1 = c_1 \log d_1$.
Let $a_i$ denote the number of Phase 1 queries to bit $i$ whose answers are $1$.
Then $a_i \sim \Bin(m_1,1-p)$ for $x_i=1$ and $a_i \sim \Bin(m_1,p)$ for $x_i=0$.
For $0\le j\le m_1$, define
\begin{align*}
  p_j = \bP(\Bin(m_1,1-p)=j) = \binom{m_1}{j} (1-p)^j p^{m_1-j}.
\end{align*}
Let $I = \left[p m_1 - m_1^{0.6}, p m_1 + m_1^{0.6}\right]$.

In Phase 2, the oracle reveals the following information:
\begin{itemize}
  \item Step 2a: The oracle reveals all $x_i$ with $a_i \notin I$.
  \item Step 2b: The oracle reveals every $x_i=1$ independently with probability $q_{a_i}$, where $q_k = 1-\frac{p_{m_1-k} p_{k_l}}{p_k p_{m_1-k_l}}$ (for integer $k \in I$) and $k_l=\lceil p m_1 - m_1^{0.6} \rceil$.
\end{itemize}

Step 2b is a reweighting process so that the posterior probability that each coordinate is $1$ does not depend on $a_i$. Intuitively, if $x_i = 0$, then the probability that $a_i = k$ for some value of $k$ and $x_i$ is unrevealed after Step 2b is $p_{m_1 - k}$. On the other hand, if $x_i = 1$, then the probability that $a_i = k$ and $x_i$ is unrevealed after Step 2b is $p_{k} \cdot (1 - q_k) = \frac{p_{m_1-k} p_{k_l}}{p_{m_1 - k_l}}$. The ratio between these two values is always the same value $\frac{p_{k_l}}{p_{m_1 - k_l}}$ regardless of the value of $k$. Thus intuitively, $a_i$ does not give information on $x_i$. 

The difference between \citet{gu2025tight}'s and our  Phase 1 and 2 is that we choose $m_1 = c_1 \log d_1$, while  \citet{gu2025tight} has $m_1 = c_1 \log n$. The analysis remains largely the same, and we only need to replace their dependency on $n$ with $d_1$. Hence, by \cite[Definition 9]{gu2025tight}, observations in Phase 1 and 2 are equivalent to the following process:
\begin{itemize}
  \item Observe every $x_i=1$ independently with probability %
  \begin{align*}
    p_+ = 1-\sum_{k\in I} p_k (1-q_k) = 1 - d_1^{-c_3\pm o(1)},
  \end{align*}
  where $c_3 = c_1 (1-2p) \log \frac{1-p}p$. By picking sufficiently small $c_1$, we can guarantee $c_3 < 1$. 
  \item Observe every $x_i=0$ independently with probability
  \begin{align*}
    p_- = \bP(\Bin(m_1,p) \notin I) = o(1).
  \end{align*}
\end{itemize}

Let $z\in \{0,1,\star\}^n$ denote the revealed information up to Phase 2.
That is, $z_i=x_i$ if $x_i$ is revealed and $z_i=\star$ otherwise.
So $\bP(z_i=1 | x_i=1) = p_+$, $\bP(z_i=0 | x_i=0) = p_-$.

For fixed $z$, let $\mu_z$ denote the distribution of $x$ conditioned on $z$.
Let $\cX_z = \supp \mu_z$. Note that $\cX_z = \{x: z_i\in \{x_i,\star\}\ \forall i\in [n]\} \cap \supp \mu$.

That is, $\mu_z$ is the posterior distribution of $x$ given the observations in Phase 1 and 2.
Let $\nu_x$ denote the distribution of $z$ conditioned on $x$ and let $\nu$ denote the overall distribution of $z$.
Note that by Bayes, 
\begin{align}
    \label{eq:nu-x-Bayes}
    \nu_x(z) = \frac{\mu_z(x) \nu(z)}{\mu(x)} \propto \frac{\mu_z(x) \nu(z)}{\deg_G(x)}. 
\end{align}

\subsection{Balance Lemma}
In this section, we establish a condition for a posterior distribution to be balanced (\Cref{lem:balance}). That is, given the observations, one cannot guess $f(x)$ with error probability $o(1)$.
As an example, we show that at the end of Phase 2, with probability $1-o(1)$, $\mu_z$ is balanced.

For $(x,y)\in E(G)$, define edge weight
\begin{align*}
  w_z(x,y) = \sqrt{\frac{\mu_z(x) \mu_z(y)}{\deg_G(x) \deg_G(y)}}.
\end{align*}
For $x\in \cX_z$, define weighted degree as %
\begin{align*}
  \wdeg_z(x) = \mu_z(x)^{-1} \sum_{y: (x,y)\in E(G)} w_z(x,y).
\end{align*}
(When $\mu_z(x)=0$, define $\wdeg_z(x)=0$.)
The weighted degree satisfies the nice property that
\begin{align*}
  \sum_{x\in f^{-1}(0)} \mu_z(x) \wdeg_z(x)
  = \sum_{x\in f^{-1}(1)} \mu_z(x) \wdeg_z(x).
\end{align*}
Define average weighted degree as
\begin{align*}
  \avgwdeg_z = \bE_{x\sim \mu_z}[\wdeg_z(x)]
   = 2\sum_{x\in f^{-1}(0)} \mu_z(x) \wdeg_z(x)
   = 2\sum_{x\in f^{-1}(1)} \mu_z(x) \wdeg_z(x).
\end{align*}

The following lemma shows that if the weighted degree is not too concentrated on a small set, then the posterior distribution is balanced. In the following, we use the convention that for any measure $\eta$ and any nonnegative function $h$, the product $\eta(A)\max_{x\in A} h(x)$
is defined to be $0$ when $A=\emptyset$.
\begin{lemma}[Balance Lemma] \label{lem:balance}
  Let $A\subseteq \cX_z$.
  Suppose
  \begin{align*}
    \max_{x\in \cX_z\backslash A} \wdeg_z(x) \le C \avgwdeg_z
  \end{align*}
  for some $C>0$.
  Then
  \begin{align*}
    \min\left\{\mu_z(f^{-1}(0)), \mu_z(f^{-1}(1))\right\} \ge \frac 1{2C} - \frac{\mu_z(A) \max_{x\in A} \wdeg_z(x)}{C \avgwdeg_z}.
  \end{align*}
\end{lemma}
\begin{proof}
  \begin{align*}
    \frac 12 \avgwdeg_z & = \sum_{x\in f^{-1}(0)} \mu_z(x) \wdeg_z(x) \\
    &\le \sum_{x\in A} \mu_z(x) \wdeg_z(x) + \sum_{x\in f^{-1}(0)\backslash A} \mu_z(x) \wdeg_z(x) \\
    &\le \mu_z(A) \max_{x\in A} \wdeg_z(x) + \mu_z(f^{-1}(0)) \max_{x\in \cX_z\backslash A} \wdeg_z(x)\\
    &\le \mu_z(A) \max_{x\in A} \wdeg_z(x) + C \avgwdeg_z \mu_z(f^{-1}(0)).
  \end{align*}
  Rearranging terms, we get
  \begin{align*}
    \mu_z(f^{-1}(0)) \ge \frac 1{2C} - \frac{\mu_z(A) \max_{x\in A} \wdeg_z(x)}{C \avgwdeg_z}.
  \end{align*}
  By the same argument, we get the bound for $\mu_z(f^{-1}(1))$.
\end{proof}

Let $x,y\in \cX_z$, $(x,y)\in E(G)$ differ at bit $i$.
Then we must have $z_i=\star$ and
\begin{align*}
  \frac{w_z(x,y)}{\mu_z(x)} &= \sqrt{ \frac{\mu_z(y)}{\mu_z(x)} \cdot \frac{1}{\deg_G(x) \deg_G(y)} } \\
  &= \frac 1{\deg_G(x)} \sqrt{ \frac{\mu_z(y)}{\mu_z(x)} \cdot \frac{\deg_G(x)}{\deg_G(y)} } \\
  &= \frac 1{\deg_G(x)} \sqrt{ \frac{\nu_y(z)}{\nu_x(z)} } \tag{\Cref{eq:nu-x-Bayes}}\\
  &= \frac 1{\deg_G(x)} \left( \frac{1-p_+}{1-p_-} \right)^{(y_i-x_i)/2},
\end{align*}
In particular,
\begin{align} \label{eqn:max-deg}
  \max_{x\in \cX_z} \wdeg_z(x) \le \sqrt{\frac{1-p_-}{1-p_+}}.
\end{align}

We now study the distribution of weighted degree.
For $x\in V(G)$, let $N_+(x) \defeq \{y:(x,y)\in E(G), |y|=|x|+1\}$ denote the up-neighbors of $x$, and $N_-(x) \defeq \{y:(x,y)\in E(G), |y|=|x|-1\}$ denote the down-neighbors of $x$.
Then for $x\in \cX_z$, we have
\begin{align}
\label{eqn:wdeg-nbrs}
  \wdeg_z(x) = \deg_G(x)^{-1} \left( |N_+(x)\cap \cX_z| \sqrt{\frac{1-p_+}{1-p_-}} + |N_-(x) \cap \cX_z| \sqrt{\frac{1-p_-}{1-p_+}} \right).
\end{align}

We have $|N_+(x)\cap \cX_z| \le |N_+(x)| \le \deg_G(x)$, so the contribution of the first term is $\le \sqrt{\frac{1-p_+}{1-p_-}}$.
For the second term, we need to prove some concentration bounds on $|N_-(x) \cap \cX_z|$.

When $z \sim \nu_x$, each down-neighbor $y$ of $x$ is in $\cX_z$ independently with probability $1-p_+$ (it happens when the corresponding bit of $x$ is not revealed, which happens with probability $1 - p_+$).
So $|N_-(x) \cap \cX_z| \sim \Bin(|N_-(x)|, 1-p_+)$.

By Chernoff bound, $|N_-(x)| \le \deg_G(x)$ and $\deg_G(x) \ge d_1$, we have
\begin{align*}
  &~\bP_{z\sim \nu_x} (|N_-(x) \cap \cX_z| \ge (1+\eps) \deg_G(x) (1-p_+) ) \\
  \le&~ \exp\left( -\frac{\eps^2}{2+\eps} \deg_G(x) (1-p_+) \right) \\
  \le&~ \exp\left( -\frac{\eps^2}{2+\eps} d_1^{1-c_3 \pm o(1)} \right),
\end{align*}
for any $\eps > 0$. 

For given $z$, let $A_{\mathrm{bad},z} = \{x\in \cX_z: |N_-(x) \cap \cX_z| \ge 2 \deg_G(x) (1-p_+)\}$.
Then (we pick $\eps = 1$, and the $\frac{\eps^2}{2+\eps}$ term can be absorbed to the $o(1)$ factor)
\begin{align}
\label{eqn:A-bad-z-bound}
  \bE_{z\sim \nu} \left[\mu_z(A_{\mathrm{bad},z})\right] = \bP_{x\sim \mu} \bP_{z\sim \nu_x} (x\in A_{\mathrm{bad},z}) &\le \exp\left( - d_1^{1-c_3 \pm o(1)} \right).
\end{align}
By Markov's inequality, we have 
\begin{align*}
  \bP_{z\sim \nu} \left(\mu_z(A_{\mathrm{bad},z}) \ge \sqrt{\exp\left( - d_1^{1-c_3 \pm o(1)} \right)}\right) \le \sqrt{\exp\left( - d_1^{1-c_3 \pm o(1)} \right)}.
\end{align*}
Note that $\sqrt{\exp\left( - d_1^{1-c_3 \pm o(1)} \right)} = \exp\left( - d_1^{1-c_3 \pm o(1)} \right)$, for suitable choices of the $o(1)$ terms. Thus, 
\begin{align} \label{eqn:muA_ub}
  \bP_{z\sim \nu} \left(\mu_z(A_{\mathrm{bad},z}) \ge \exp\left( - d_1^{1-c_3 \pm o(1)} \right)\right) \le \exp\left( - d_1^{1-c_3 \pm o(1)} \right).
\end{align}

For $x\in \cX_z\backslash A_{\mathrm{bad},z}$, we have
\begin{align} \label{eqn:wdeg_z_x_Ac}
  \nonumber \wdeg_z(x) &\le \deg_G(x)^{-1} \left( \deg_G(x) \sqrt{\frac{1-p_+}{1-p_-}} + 2 \deg_G(x) (1-p_+) \sqrt{\frac{1-p_-}{1-p_+}} \right) \\
  &\le 4 \sqrt{(1-p_+)(1-p_-)}.
\end{align}

We now have bounds on $\mu_z(A_{\mathrm{bad},z})$ and $\wdeg_z(x)$ for $x\in \cX_z\backslash A_{\mathrm{bad},z}$.
Next, we establish a lower bound on $\avgwdeg_z$.

Fix $x\in V(G)$.
Then $\max\{|N_+(x)|, |N_-(x)|\} \ge \deg_G(x)/2 \ge d_1/2$.
If $|N_-(x)| \ge \deg_G(x)/2$, then
\begin{align*} 
  \bP_{z\sim \nu_x} \left(|N_-(x) \cap \cX_z| \le \frac 14 \deg_G(x) (1-p_+)\right) \le \exp\left( -d_1^{1-c_3 \pm o(1)} \right).
\end{align*}
Similar to $|N_-(x) \cap \cX_z|$, $|N_+(x) \cap \cX_z| \sim \Bin(|N_+(x)|, 1-p_-)$. So if $|N_+(x)| \ge \deg_G(x)/2$, then
\begin{align} 
\label{eqn:conc-n-plus}
  \bP_{z\sim \nu_x} \left(|N_+(x) \cap \cX_z| \le \frac 14 \deg_G(x) (1-p_-)\right) \le
  \exp\left( -d_1^{1-o(1)}\right).
\end{align}
Combining the two cases and recalling \Cref{eqn:wdeg-nbrs}, we get
\begin{align*}
  &~\bP_{z\sim \nu} \bP_{x\sim \mu_z} \left( \wdeg_z(x) \le \frac 14 \sqrt{(1-p_+)(1-p_-)} \right) \\
  =&~ \bP_{x\sim \mu} \bP_{z\sim \nu_x} \left( \wdeg_z(x) \le \frac 14 \sqrt{(1-p_+)(1-p_-)} \right) \\
  \le&~ \exp\left( -d_1^{1-c_3 \pm o(1)} \right).
\end{align*}
By Markov's inequality, we have
\begin{align*}
  \bP_{z\sim \nu} \left(\bP_{x\sim \mu_z} \left( \wdeg_z(x) \le \frac 14 \sqrt{(1-p_+)(1-p_-)} \right) > 1/2 \right) \le \exp\left( -d_1^{1-c_3 \pm o(1)}\right).
\end{align*}
So
\begin{align} \label{eqn:avgwdeg_z_lb}
  \bP_{z\sim \nu} \left(\avgwdeg_z \le \frac 18 \sqrt{(1-p_+)(1-p_-)} \right) \le \exp\left( -d_1^{1-c_3 \pm o(1)}\right).
\end{align}

Let $\cE_1$ be the event that $\mu_z(A_{\mathrm{bad},z}) \le \exp\left( - d_1^{1-c_3 \pm o(1)} \right)$ (recall \Cref{eqn:muA_ub}).
Let $\cE_2$ be the event that $\avgwdeg_z \ge \frac 18 \sqrt{(1-p_+)(1-p_-)}$.
Then by \Cref{eqn:muA_ub,eqn:avgwdeg_z_lb}, we have
\begin{align}
  \bP_{z\sim \nu} (\cE_1 \cap \cE_2) \ge 1-\exp\left( -d_1^{1-c_3 \pm o(1)}\right).
\end{align}

Conditioned on $\cE_1 \cap \cE_2$, we apply \Cref{lem:balance}.
By \Cref{eqn:wdeg_z_x_Ac} and definition of $\cE_2$, we can take $C=32$ in \Cref{lem:balance}.
So we have
\begin{align*}
  \min\left\{\mu_z(f^{-1}(0)), \mu_z(f^{-1}(1))\right\} &\ge \frac 1{64} - \frac{\mu_z(A_{\mathrm{bad},z}) \max_{x\in A_{\mathrm{bad},z}} \wdeg_z(x)}{32 \avgwdeg_z} \\
  &\ge \frac 1{64} - \frac{\exp\left( - d_1^{1-c_3 \pm o(1)} \right) \sqrt{\frac{1-p_-}{1-p_+}}}{32 \cdot \frac 18 \sqrt{(1-p_+)(1-p_-)}} \\
  &\ge \frac 1{64} - o(1).
\end{align*}

This means that with probability at least $1-\exp\left( -d_1^{1-c_3 \pm o(1)}\right)$ over $z$, the posterior distribution $\mu_z$ is balanced after Phase 2.

\subsection{Phase 3}
In Phase 3, the algorithm adaptively queries $m_2 = c_2 d_2$ bits in expectation.
We will show that with constant probability, after the $m_2$ queries, $\avgwdeg_z$ decreases by at most a constant factor. Then we can use \Cref{lem:balance} again to show that the posterior distribution is still balanced.

Let $B = \{x\in V(G): \deg_G(x) \ge d_2\}$. By the definition of $d_2$, $\mu(B) \ge \gamma$.
For every $x\in B$, we have $\max\{|N_+(x)|, |N_-(x)|\} \ge \deg_G(x)/2 \ge d_2/2$.
Let $B_+ = \{x\in B: |N_+(x)| \ge \deg_G(x)/2\}$ and $B_- = \{x\in B: |N_-(x)| \ge \deg_G(x)/2\}$.
Then $\max\{\mu(B_+), \mu(B_-)\} \ge \mu(B)/2 \ge \gamma/2$.
WLOG assume that $\mu(B_+) \ge \gamma/2$.
(If not, we can define $g(x) = f(\mathbf{1}-x)$ and prove a lower bound for $g$ instead. Clearly, $\N_p(f) = \N_p(g)$. )

We show that after Phase 2, with constant probability, vertices in $B_+$ contribute a constant amount to $\avgwdeg_z$.
Let $B_{\mathrm{good},z} = \{x\in B_+: |N_+(x) \cap \cX_z| \ge \frac 14 \deg_G(x) (1-p_-)\}$. %
For $x\in B_+$, by the same analysis as \Cref{eqn:conc-n-plus} (and now we have $\deg_G(x) \ge d_2$,  instead of the bound $\deg_G(x) \ge d_1$ used in \Cref{eqn:conc-n-plus}),
\begin{align*}
  \bP_{z\sim \nu_x} \left( x\in B_{\mathrm{good},z} \right) \ge 1-\exp\left( -d_2^{1-c_3 \pm o(1)} \right) = 1-o(1).
\end{align*}

So
\begin{align*}
  \bE_{z\sim \nu} \mu_z(B_{\mathrm{good},z})
  = \bE_{x\sim \mu} \left[\mathbbm{1}_{x\in B_+} \bP_{z\sim \nu_x} \left( x\in B_{\mathrm{good},z} \right) \right]
  \ge \gamma/2-o(1).
\end{align*}

By \Cref{cor:markov}, we have 
\begin{align*}
  \bP_{z\sim \nu} \left(\mu_z(B_{\mathrm{good},z}) \ge \gamma/4\right) \ge \gamma/4-o(1).
\end{align*}
Let $\cE_3$ be the event that after Phase 2, $\mu_z(B_{\mathrm{good},z}) \ge \gamma/4$.
Then $\bP_{z\sim \nu} (\cE_3) \ge \gamma/4-o(1)$.

For any $k$, define $\cE_{\le k} = \bigcap_{j=1}^k \cE_j$.
Then $\bP_{z\sim \nu} (\cE_{\le 3}) \ge \gamma/4 - o(1)$ by union bound. 

Note that for $x\in B_{\mathrm{good},z}$,  by \Cref{eqn:wdeg-nbrs}, we have
\begin{align*}
  \wdeg_z(x) \ge \frac 14 \sqrt{(1-p_+)(1-p_-)}.
\end{align*}
So at the beginning of Phase 3, conditioned on $\cE_{\le 3}$, we have
\begin{align*}
  \avgwdeg_z \ge \sum_{x\in B_{\mathrm{good},z}} \mu_z(x) \wdeg_z(x) \ge \frac{\gamma}{16} \sqrt{(1-p_+)(1-p_-)}.
\end{align*}
This bound is weaker than \Cref{eqn:avgwdeg_z_lb}, but the analysis is more robust under queries in Phase 3.

In the following, we will study the effect of queries in Phase 3.
Let $q$ be the number of queries in Phase 3. Note that $\bE[q] = m_2$.
Let $i_1,\ldots,i_q$ denote the queried bits in Phase 3 and $b_1,\ldots,b_q$ denote the corresponding query results.
For $0\le t\le q$, let $\mu_{z,t}$ denote the posterior distribution after $t$ queries in Phase 3. Then $\mu_{z,0} = \mu_z$.

For the $t$-th query, the query result $b_t$ is equal to $b\in \{0,1\}$ with probability $\bP_{x\sim \mu_{z,t-1}} (x_{i_t}=b)$.
For a fixed $x$, $\mu_{z,t}(x)$ changes to $0$ if $x_{i_t} \ne b_t$ and changes to $\frac{\mu_{z,t-1}(x)}{\bP_{x'\sim \mu_{z,t-1}} (x'_{i_t}=b_t)}$ if $x_{i_t} = b_t$.
The expectation is equal to $\mu_{z,t-1}(x)$.
Therefore, $\mu_{z,t}(x)$ is a martingale over the query results.
In particular, $\bE[\mu_{z,q} \mid z] = \mu_z$ by the optional stopping theorem.

Conditioned on $\cE_{\le 3}$,
\begin{align*}
  \bE\left[\mu_{z,q}(B_{\mathrm{good},z})\right] = \bE\left[\mu_z(B_{\mathrm{good},z})\right] \ge \gamma/4.
\end{align*}
By \Cref{cor:markov}, we have
\begin{align*}
  \bP \left( \mu_{z,q}(B_{\mathrm{good},z}) \ge \gamma/8 \right) \ge \gamma/8.
\end{align*}
Let $\cE_4$ be the event that $\mu_{z,q}(B_{\mathrm{good},z}) \ge \gamma/8$.
Then $\bP(\cE_4 \mid \cE_{\le 3}) \ge \gamma/8$.

Because $\bE[\mu_{z,q} \mid z] = \mu_z$, we have $\bE[\mu_{z,q}(A_{\mathrm{bad},z})] = \bE[\mu_z(A_{\mathrm{bad},z})] \le \exp\left( - d_1^{1-c_3 \pm o(1)} \right)$ (recall \Cref{eqn:A-bad-z-bound}).
By Markov's inequality, we have (similar to \Cref{eqn:muA_ub})
\begin{align*}
  \bP\left( \mu_{z,q}(A_{\mathrm{bad},z}) \ge \exp\left( - d_1^{1-c_3 \pm o(1)} \right) \right) \le \exp\left( - d_1^{1-c_3 \pm o(1)} \right).
\end{align*}
Let $\cE_5$ be the event that $\mu_{z,q}(A_{\mathrm{bad},z}) \le \exp\left( - d_1^{1-c_3 \pm o(1)} \right)$.
Then 
\[\bP(\cE_5 \mid \cE_{\le 3}) \ge 1 - \bP(\neg \cE_5 \mid \cE_{\le 3}) \ge 1 - \bP(\neg \cE_5) / \bP(\cE_{\le 3}) \ge 1-o(1).\]
By union bound, $\bP(\cE_{\le 5}) = \bP(\cE_{\le 3}) \bP(\cE_4 \cap \cE_5 \mid \cE_{\le 3}) \ge (\gamma/4 - o(1)) (\gamma/8 - o(1)) = \gamma^2/32 - o(1)$.

Because $\bE[q] = m_2$, we have
\begin{align*}
  \bE[q \mid \cE_{\le 5}] \le \frac{m_2}{\bP(\cE_{\le 5})} \le (1+o(1))\frac{32m_2}{\gamma^2}.
\end{align*}
So conditioned on $\cE_{\le 5}$, with probability at least $1/2-o(1)$, we have $q\le \frac{64m_2}{\gamma^2} = \frac{64 c_2 d_2}{\gamma^2}$. Let $\cE_6$ be the event that $q\le \frac{64 c_2 d_2}{\gamma^2}$.
Then $\bP(\cE_{\le 6}) \ge \gamma^2/64 - o(1)$.

From now on we condition on $\cE_{\le 6}$.
We would like to show that $\mu_{z,q}$ is balanced.
For $0\le t\le q$, define $\cX_{z,t} = \supp \mu_{z,t}$ and
\begin{align*}
  w_{z,t}(x,y) &= \sqrt{\frac{\mu_{z,t}(x) \mu_{z,t}(y)}{\deg_G(x) \deg_G(y)}}, \\
  \wdeg_{z,t}(x) &= \mu_{z,t}(x)^{-1} \sum_{y: (x,y)\in E(G)} w_{z,t}(x,y),\\
  \avgwdeg_{z,t} &= \bE_{x\sim \mu_{z,t}}[\wdeg_{z,t}(x)].
\end{align*}
We have the following variation of \Cref{lem:balance} for $\mu_{z,q}$.
\begin{lemma}[Balance Lemma] \label{lem:balance-mu-z-q}
  Let $A\subseteq \cX_{z,q}$.
  Suppose
  \begin{align*}
    \max_{x\in \cX_{z,q}\backslash A} \wdeg_{z,q}(x) \le C \avgwdeg_{z,q}
  \end{align*}
  for some $C>0$.
  Then
  \begin{align*}
    \min\left\{\mu_{z,q}(f^{-1}(0)), \mu_{z,q}(f^{-1}(1))\right\} \ge \frac 1{2C} - \frac{\mu_{z,q}(A) \max_{x\in A} \wdeg_{z,q}(x)}{C \avgwdeg_{z,q}}.
  \end{align*}
\end{lemma}
The proof is essentially the same as \Cref{lem:balance} and is omitted.

For $x\in \cX_{z,t}$, we have
\begin{align*}
  \wdeg_{z,t}(x) &= \sum_{y: (x,y)\in E(G)} \sqrt{\frac{\mu_{z,t}(y)}{\mu_{z,t}(x)} \cdot \frac{1}{\deg_G(x) \deg_G(y)} } \\
  & = \sum_{y: (x,y)\in E(G)} \mathbbm{1}_{y\in \cX_{z,t}} \sqrt{\frac{\mu_{z}(y)}{\mu_{z}(x)} \cdot \frac{1}{\deg_G(x) \deg_G(y)} } \\
  & = \deg_G(x)^{-1} \left( |N_+(x)\cap \cX_{z,t}| \sqrt{\frac{1-p_+}{1-p_-}} + |N_-(x) \cap \cX_{z,t}| \sqrt{\frac{1-p_-}{1-p_+}} \right).
\end{align*}
For any $x\in B_{\mathrm{good},z} \cap \cX_{z,q}$, $|N_+(x) \cap \cX_{z,t}|$ decreases by at most one after each query.
So
\begin{align}
  |N_+(x) \cap \cX_{z,q}| \ge |N_+(x) \cap \cX_{z}| - q \ge \frac 14 \deg_G(x) (1-p_-) - \frac{64 c_2 d_2}{\gamma^2}.
\end{align}
Because $\deg_G(x) \ge d_2$, for $c_2=\frac{1}{1024}\gamma^2$, the last expression is at least $\frac 18 \deg_G(x) (1-p_-)$ for sufficiently large $d_1$ as $p_- = o(1)$.
Then
\begin{align*}
  \avgwdeg_{z,q} &\ge \sum_{x\in B_{\mathrm{good},z} \cap \cX_{z,q}} \mu_{z,q}(x) \wdeg_{z,q}(x) \\
  &\ge \mu_{z,q}(B_{\mathrm{good},z}) \cdot \frac 18 \sqrt{(1-p_+)(1-p_-)} \tag{by definition of $B_{\mathrm{good},z}$}\\
  &\ge \frac{\gamma}{64} \sqrt{(1-p_+)(1-p_-)} \tag{by definition of $\cE_4$}.
\end{align*}

On the other hand $\wdeg_{z,q}(x) \le \wdeg_z(x)$ for all $x\in \cX_{z,q}$.
So
\begin{align*}
  \max_{x\in \cX_{z,q}\backslash A_{\mathrm{bad},z}} \wdeg_{z,q}(x) \le \max_{x\in \cX_z \backslash A_{\mathrm{bad},z}} \wdeg_z(x) \overset{\Cref{eqn:wdeg_z_x_Ac}}\le 4 \sqrt{(1-p_+)(1-p_-)} \le \frac{256}{\gamma} \avgwdeg_{z,q}.
\end{align*}
Finally, let $A \defeq A_{\mathrm{bad},z} \cap \cX_{z,q}$, then
\begin{align*}
  &~\frac{\mu_{z,q}(A) \max_{x\in A} \wdeg_{z,q}(x)}{\avgwdeg_{z,q}}\\
  \le&~  \frac{\mu_{z,q}(A) \max_{x\in A} \wdeg_{z}(x)}{\frac{\gamma}{64} \sqrt{(1-p_+)(1-p_-)}} \tag{by definition of $\cE_5$}\\
  \le&~  \frac{\exp\left( - d_1^{1-c_3 \pm o(1)} \right) \sqrt{\frac{1-p_-}{1-p_+}}}{\frac{\gamma}{64} \sqrt{(1-p_+)(1-p_-)}} \tag{\Cref{eqn:max-deg}}\\
  =&~ o(1).
\end{align*}
We now apply \Cref{lem:balance-mu-z-q} with $C = \frac{256}{\gamma}$ and $A$ to get that $\mu_{z,q}$ is balanced.
That is, the posterior distribution $\mu_{z,q}$ satisfies
\begin{align*}
  \min\left\{\mu_{z,q}(f^{-1}(0)), \mu_{z,q}(f^{-1}(1))\right\} \ge \frac{\gamma}{512} - o(1)
\end{align*}
conditioned on $\cE_{\le 6}$. Recall $\bP(\cE_{\le 6}) \ge \gamma^2 / 64 - o(1)$. 
Therefore, for small enough $\delta>0$, any algorithm cannot compute $f(x)$ with error probability at most $\delta$ in the three-phase problem.
This finishes the proof.

\section{Applications} \label{sec:app}
In this section, we show various applications of \Cref{thm:main-general}. We start with the following lemma which will be used in several of the proofs.

\begin{lemma}
\label{lem:many-high-deg-vertices}
    Let $\lambda \in [0, 1), \gamma_1, \gamma_2, \gamma_3 > 0, p \in (0, 1 / 2)$ be constants. Then there exists $c = c(\lambda, \gamma_1, \gamma_2, \gamma_3, p)$ such that the following holds. 
    Let $f: \{0, 1\}^n \rightarrow \{0, 1\}$ be a Boolean function and let $G$ be an induced subgraph of $H_f$. Suppose there exists a set $S$ of $\gamma_1 n^{-\lambda} \cdot |V(G)|$ vertices with degree at least $\gamma_2 n$ in $G$, and the total number of edges in $G$ is at most $\gamma_3 n |S|$.  Then $\N_p(f) \ge c n \log n$. 
\end{lemma}
\begin{proof}
    We iteratively remove vertices whose degrees are smaller than $\frac{1}{4} \gamma_1 \gamma_2 n^{1-\lambda}$ from $G$ until there are no such vertices left. Let $G'$ denote the remaining graph and let $S'$ denote $V(G') \cap S$. 

    The total number of edges we delete from $G$ is at most $\frac{1}{4} \gamma_1 \gamma_2 n^{1-\lambda} \cdot |V(G)|$. Because initially, the total degree of all vertices from $S$ is at least  $|S| \cdot \gamma_2 n$, the total degree of all vertices from $S'$ in $G'$ is at least (removing one edge can decrease the degrees of two vertices)
    \begin{align*}
    |S| \cdot \gamma_2 n - 2 \cdot \frac{1}{4} \gamma_1 \gamma_2 n^{1-\lambda} \cdot |V(G)| \ge \frac{1}{2} |S| \cdot \gamma_2 n. 
    \end{align*}
    Hence, if we sample a random $s \in S$, $\bE[\deg_{G'}(s) / n] \ge \frac{1}{2} |S| \cdot \gamma_2 n / |S| / n \ge \gamma_2 / 2$ (here, if $s \not \in S'$, we regard $\deg_{G'}(s)$ as $0$). By \Cref{cor:markov}, $\Pr(\deg_{G'}(s) / n \ge \gamma_2 / 4) \ge \gamma_2 / 4$, i.e., there are $\frac{1}{4} \gamma_2 |S|$ vertices whose degrees are at least $\frac{1}{4} \gamma_2 n$ in $G'$. We call these vertices $S_{\text{good}}$. 
    
    Let $\mu$ be the distribution from \Cref{thm:main-general}'s statement applied to $G'$, i.e., $\mu(x) = \frac{\deg_{G'}(x)}{2|E(G')|}$ for $x \in V(G')$. We have
    \begin{align*}
        \mu(S_{\text{good}}) & = \frac{\sum_{s \in S_{\text{good}}} \deg_{G'}(s)}{2|E(G')|}\\
        & \ge \frac{\left(\frac{1}{4} \gamma_2 |S|\right) \cdot \left(\frac{1}{4} \gamma_2 n\right)}{2|E(G)|}\\
        & \ge \frac{\left(\frac{1}{4} \gamma_2 |S|\right) \cdot \left(\frac{1}{4} \gamma_2 n\right)}{2 \cdot \gamma_3 n |S|}\\
        & \ge \frac{\gamma_2^2}{32 \gamma_3}. 
    \end{align*}
    Therefore, we can apply \Cref{thm:main-general} with $d_1 \ge \frac{1}{4} \gamma_1 \gamma_2 n^{1-\lambda}$ and $d_2 \ge \frac{1}{4} \gamma_2 n$ to obtain 
    \[
    \N_p(f) \ge c \frac{1}{4} \gamma_2 n \cdot \log\left(1 + \frac{1}{4} \gamma_1 \gamma_2 n^{1-\lambda}\right) \ge c' n \log n
    \] 
    for some constant $c'$. 
\end{proof}

Next, we prove the lower bound for biased random functions. 
\Biased*
\begin{proof}
First, by Chernoff bound, the number of $x$ where $f(x) = 1$ is between $\frac{2}{3} q \cdot 2^n$ and $\frac{4}{3} q \cdot 2^n$ with probability $\ge 1 - \exp(-c_1 q 2^n)$ for some constant $c_1 > 0$. 
Let this event be $\cE_1$. 

    Next, for any fixed $x$, consider its neighbors in the Boolean hypercube. The number of them that evaluate to $1$ is distributed as $\Bin(n, q)$. Let us denote the random variable by $Y$. We assume $n$ is sufficiently large so that $q \le 1/2$. Then by Chernoff bound, 
    \begin{align*}
        \Pr\left(Y \ge \frac{2}{3} n\right) \le \exp(-c_2 n)
    \end{align*}
    for some constant $c_2 > 0$. Let $X_{\text{bad}} \subseteq \{0, 1\}^n$ be the set of $x$ whose neighborhood contains more than $\frac{2}{3} n$ elements that evaluate to $1$. Hence, $\bE\left[|X_{\text{bad}}|\right] \le 2^n \cdot \exp(-c_2n)$. By Markov's inequality, $\Pr(|X_{\text{bad}}| > 2^n \cdot \exp(-c_2n / 2)) \le \exp(-c_2n / 2)$. Let $\cE_2$ denote the event where $|X_{\text{bad}}| \le 2^n \cdot \exp(-c_2n / 2)$. 

    We will show that under $\cE_1 \cap \cE_2$ (which happens with probability $1-\exp(-c n)$ for some constant $c$ by union bound), $\N_p(f) \ge \Omega(n \log n)$. 

    For any $x \in f^{-1}(1) \setminus X_{\text{bad}}$, it has at least $n/3$ neighbors in $H_f$. Furthermore, under $\cE_1 \cap \cE_2$, $\left| f^{-1}(1) \setminus X_{\text{bad}} \right| \ge \frac{2}{3} q \cdot 2^n - |X_{\text{bad}}| \ge (\frac{2}{3} - o(1)) q \cdot 2^n$. Finally, under $\cE_1 \cap \cE_2$, the total number of edges in $H_f$ is at most $\frac{4}{3} q \cdot n\cdot 2^n \le (2+o(1)) n | f^{-1}(1) \setminus X_{\text{bad}}|$.
    
    Therefore, we can apply \Cref{lem:many-high-deg-vertices} where $G \defeq H_f$ and $S \defeq f^{-1}(1) \setminus X_{\text{bad}}$. This finishes the proof. 
\end{proof}

Next, we show our lower bound for the subsequence counting problem: 
\Subseq*
\begin{proof}
    Assume without loss of generality that $t_k = 1$, by complementing all bits if $t_k = 0$. Also, assume without loss of generality $t_j = 0$ for some $j \in [k]$, as otherwise the lower bound is implied by the known lower bound of computing $|x|$ \citep{feige1994computing}. 

    Let $C$ be a sufficiently large constant. Let $\ell = \lfloor n / (k + (C-1)(k - |t|)) \rfloor$ and $r = n - \ell ((k + (C-1)(k - |t|)))$. 
    We will consider a subset of all possible inputs, where all but the last $\ell$ bits are fixed. More precisely, let $L_1$ be the concatenation of $\ell$ bits of $1$'s, and let $L_0$ be the concatenation of $C\ell$ bits of $0$'s. Let $P$ be $r$ copies of the bit $1 - t_1$. For $x \in \{0, 1\}^\ell$, define 
    \[
    g(x) = f(P, L_{t_1}, L_{t_2}, \ldots, L_{t_{k-1}}, x). 
    \]
    Note that the first $r$ bits cannot be involved in any subsequence $t$. 
    Clearly, an $\Omega(\ell \log \ell)$ lower bound for computing $g$ will imply an $\Omega(\ell \log \ell) = \Omega(n \log n)$ lower bound for computing $f$. 

    For any $x_i = 1$, the number of subsequences $t$ using $x_i$ in $L_{t_1}, L_{t_2}, \ldots, L_{t_{k-1}}, x$ is at least $\prod_{j=1}^{k-1} |L_{t_j}| = C^{k - |t|} \cdot \ell^{k-1}$. On the other hand, for any $x_i = 0$, the number of subsequences involving $x_i$ is at most $k^{k-1} C^{k-|t| - 1} \ell^{k-1}$ (by using the fact that the number of $1$'s in the whole sequence is at most $k\ell$ and the number of $0$'s in the whole sequence is at most $Ck\ell$). Hence, by setting $C > 2k^{k-1}$, we can guarantee that $g(y) - g(x) \ge c_1 \cdot \ell^{k-1}$ for some constant $c_1 > 0$, where $y$ and $x$ differ at exactly one bit and $y_i = 1, x_i = 0$ on that bit. 

    On the other hand, if $y$ and $x$ differ at one bit, we have $|g(y) - g(x)| = O(\ell^{k-1})$, because each bit can only be involved in $O(\ell^{k-1})$ subsequences. Hence, we can apply the  Azuma-Hoeffding inequality to obtain that, 
    \begin{align}
    \label{eqn:subseq:azuma}
    \bP_{x \sim \{0, 1\}^\ell}\left(|g(x) - \bE[g(x)]| \le c_2 \ell^{k-0.5} \right) \ge 0.9
    \end{align}
    for some constant $c_2 > 0$.

    Furthermore, let $B_{\text{bad}}$ denote the set of $x \in \{0, 1\}^\ell$ where the number of $0$'s is less than $\ell / 3$. By Chernoff bound, $|B_{\text{bad}}| \le \exp(-c_3 \ell) \cdot 2^\ell$ for some constant $c_3 > 0$.  
    
    Next, we split $\bZ_{\ge 0}$ into intervals of length $c_1 \cdot \ell^{k-1}$, and we denote the intervals by $I_1, I_2, \ldots.$ Let $I_{m}$ be the interval maximizing $|g^{-1}(I_m)|$. Let $S$ denote $g^{-1}(I_m) \setminus B_{\text{bad}}$. Let $T$ be the largest element in $I_m$, and consider the threshold version $g_T$ of $g$ defined as $g_T(x) = [g(x) \le T]$. 

    \begin{itemize}
        \item By \Cref{eqn:subseq:azuma}, there exist $O(\sqrt{\ell})$ intervals $I_a, I_{a+1}, \ldots, I_b$ such that $\sum_{i=a}^b |g^{-1}(I_i)| \ge 0.9 \cdot 2^\ell$. Hence, $|g^{-1}(I_m)| \ge 0.9 \cdot 2^\ell / O(\sqrt{\ell}) = \Omega(2^\ell / \sqrt{\ell})$ as $m$ maximizes $|g^{-1}(I_m)|$. By the upper bound on $B_{\text{bad}}$, this also implies $|S| \ge |g^{-1}(I_m)| - \exp(-c_3 \ell) \cdot 2^\ell = \Omega(|g^{-1}(I_m)|) = \Omega(2^\ell / \sqrt{\ell})$.
        \item For every $x \in S$, if we flip any $x_i = 0$ to $1$ and obtain $y$, $g(y) - g(x) \ge c_1 \cdot \ell^{k-1}$ by previous discussion. By definition of intervals and $T$, $T - c_1 \ell^{k-1} < g(x) \le T$. Hence, $g(y) > T$. This means that $g_T(x) \ne g_T(y)$, and $y$ is a neighbor of $x$ in $H_{g_T}$. As $x \notin B_{\text{bad}}$, the number of such neighbors is $\Omega(\ell)$. 
        \item Finally, let $(x, y)$ be any edge in $H_{g_T}$, with $g_T(x) = 1$ and $g_T(y) = 0$, i.e., $g(x) \le T$ and $g(y) > T$. As discussed earlier, $g(y) - g(x) = O(\ell^{k-1})$, so we must have $g(x) > T - O(\ell^{k-1})$. This means that $x$ is from $g^{-1}(I_{i})$ for some $i = m - O(1)$. Hence, the total number of edges of $H_{g_T}$ is upper bounded by 
        \begin{align*}
            \sum_{i = m - O(1)}^m \ell \cdot |g^{-1}(I_{i})| = O(\ell |g^{-1}(I_m)|) = O(\ell |S|). 
        \end{align*}
    \end{itemize}
    Because of the three bullet points, we can apply \Cref{lem:many-high-deg-vertices} with input dimension $\ell$, $\lambda = 1/2$,
    $G = H_{g_T}$ and the above set $S$ to obtain $\N_p(g_T) = \Omega(\ell \log \ell) = \Omega(n \log n)$, which implies $\N_p(g) = \Omega(n \log n)$ and $\N_p(f) = \Omega(n \log n)$.
\end{proof}

\subsection{New Proofs of Previous Results}

In this section, we show that our main technical theorem is general enough to imply all previous known nontrivial lower bounds on the noisy query complexity. As discussed in \Cref{sec:intro}, it suffices to consider \Threshold{k}, symmetric functions, \Conn{} and \stConn{}. 

We first consider symmetric functions. 
A Boolean function $f: \{0, 1\}^n \rightarrow \{0, 1\}$ is symmetric if for any permutation $\pi$ on $[n]$ and any $x \in \{0, 1\}^n$, $f(x) = f(x_{\pi(1)}, x_{\pi(2)}, \ldots, x_{\pi(n)})$. In other words, the value of $f(x)$ only depends on the Hamming weight of $x$.

\begin{theorem}[\cite{feige1994computing}]
\label{thm:symmetric}
Let $f$ be a symmetric function. Let $1 \le k \le n$ be any integer such that there exists $y_0$ with Hamming weight $k - 1$ and $y_1$ with Hamming weight $k$ such that $f(y_0) \ne f(y_1)$. Then the noisy query complexity of $f$ is $\Omega(n \log \min\{k, n - k + 1\})$.
\end{theorem}
\begin{proof}
For every $0 \le i \le n$, let $W_i$ be the set of all $x \in \{0,1\}^n$ with Hamming weight $i$. Consider the induced subgraph of $H_f$ on vertices $W_{k-1} \cup W_k$. Because $f$ is symmetric, $f(x) = f(y_0)$ for every $x \in W_{k-1}$ and $f(x) = f(y_1)$ for every $x \in W_k$. Hence, the degree of every vertex $x \in W_{k-1}$ in this induced subgraph is $n - k + 1$, and the degree of every $x \in W_k$ is $k$. The minimum degree is thus $d_1 \ge \min\{k, n - k + 1\}$. Let $\mu(x)$ be the distribution defined in \Cref{thm:main-general} where $\mu(x) = \frac{\deg_G(x)}{2 |E(G)|}$. Note that
\[
\Pr_{x \sim \mu}[x \in W_{k}] = \frac{|W_k| \cdot k}{2 |W_{k}| \cdot k} = 1/2.
\]
Therefore, a random $x$ sampled with respect to distribution $\mu$ has degree $k$ with $1/2$ probability, and has degree $n - k + 1$ with $1/2$ probability. Hence, letting $d_2 = \max\{k, n - k + 1\}$
\[
\Pr_{x \sim \mu}[\deg_G(x) \ge d_2] \ge 1/2.
\]
Therefore, \Cref{thm:main-general} implies that $\N_p(f) = \Omega(d_2 \log d_1) = \Omega(\max\{k, n - k + 1\} \log \min\{k, n - k + 1\}) = \Omega(n \log \min\{k, n - k + 1\})$.
\end{proof}
In the \Threshold{k} problem where $1 \le k \le n$, we are given $x \in \{0, 1\}^n$, and we need to determine whether $|x| \ge k$. As \Threshold{k} is a special case of symmetric functions, \Cref{thm:symmetric} immediately implies an $\Omega(n \log \min\{k, n - k + 1\})$ lower bound for the noisy query complexity of \Threshold{k}, which was originally shown in \citet{feige1994computing}.

In \Conn{}  for $N$-vertex graphs, the input is a Boolean vector $x \in \{0, 1\}^{\binom{[N]}{2}}$, where each $x_{\{u, v\}}$ for an unordered pair $\{u, v\}$ indicates whether the edge $(u, v)$ exists in the graph. The goal is to determine whether the graph is connected.

\begin{theorem}[\cite{gu2025tight}]
    The noisy query complexity for \Conn{}  for $N$-vertex graphs is $\Omega(N^2 \log N)$.
\end{theorem}
\begin{proof}
    Let $f$ denote \Conn{} throughout this proof. 

    Let $S_1$ denote the set of all spanning trees on $N$ vertices, and let $S_0$ denote the set of all disjoint union of two spanning trees $\{T_1, T_2\}$ ($T_1$ and $T_2$ are not ordered) where $|T_1| + |T_2| = N$ and $|T_1|, |T_2| \ge N / 3$. Clearly, $f(x) = 1$ for $x \in S_1$ and $f(x) = 0$ for $x \in S_0$. We use the following facts to estimate the sizes of $S_0$ and $S_1$:
    \begin{fact}
        The number of labeled $N$-vertex trees is $N^{N-2}$.
    \end{fact}
     \begin{fact}[Stirling's formula]
        There is a constant $\eps > 0$, so that for all $N \ge k \ge 1$, 
        \[
        \binom{N}{k} \ge \eps \cdot \frac{N^{N+1/2}}{k^{k+1/2}(N-k)^{N-k+1/2}}. 
        \]
    \end{fact}
    We immediately obtain $|S_1| = N^{N-2}$ and 
    \begin{align*}
        |S_0| &\ge \sum_{k=\lceil N/3 \rceil}^{\lfloor N/2-1 \rfloor} \binom{N}{k}k^{k-2} (N-k)^{N-k-2}\\      
        & \ge \sum_{k=\lceil N/3 \rceil}^{\lfloor N/2-1 \rfloor} \Omega\left(\frac{N^{N+1/2}}{k^{k+1/2} (N-k)^{N-k+1/2}} \cdot k^{k-2} (N-k)^{N-k-2}\right)\\
        & \ge \sum_{k=\lceil N/3 \rceil}^{\lfloor N/2-1 \rfloor} \Omega\left(N^{N-4.5}\right)\\
        & \ge \Omega(N^{N-3.5}) \ge \Omega(|S_1| / N^{1.5}). 
    \end{align*}

    Consider the induced subgraph of $H_f$ on the vertex set $S_0 \cup S_1$, which we call $G$. For every $x = T_1 \sqcup T_2 \in S_0$, $\deg_G(x) \ge \frac{2N^2}{9}$, because we can connect any edge between $T_1$ and $T_2$ to obtain a spanning tree. Because $G$ is bipartite where $S_0$ is one side of the bipartition,
    \begin{align*}
        |E(G)| \le N^2 \cdot |S_0|.
    \end{align*}

    Hence, we have obtained a subgraph $G$ of $H_f$ where there exists $S_0 \subseteq V(G)$ such that 
    \begin{itemize}
        \item $\deg_G(s) = \Omega(N^2)$ for every $s \in S_0$;
        \item $|S_0| \ge \Omega(N^{-1.5}) |V(G)|$;
        \item $|E(G)| \le N^2 |S_0|$.
    \end{itemize}
    Hence, we can apply \Cref{lem:many-high-deg-vertices} with $n\defeq \binom{N}{2}$, $S \defeq S_0$ and $\lambda = 3/4$ to obtain the $\Omega(N^2 \log N)$ lower bound.

\end{proof}

In the related \stConn{} problem, instead of determining whether a given graph is connected, we need to determine whether two fixed vertices $s$ and $t$ in the graph are connected. Via a similar proof, we can also show the lower bound for \stConn{}.

\begin{theorem}[\cite{gu2025tight}]
    The noisy query complexity for \stConn{}  for $N$-vertex graphs is $\Omega(N^2 \log N)$.
\end{theorem}
\begin{proofsketch}
    As the proof is very similar to the proof of the lower bound for \Conn{}, we only highlight the difference here. The initial subgraph of $H_f$ we will look at is on $S_0 \cup S_1$, where $S_1$ still denotes the set of all spanning trees, but $S_0$ denotes the set of all disjoint union of two spanning trees $\{T_1, T_2\}$ where $|T_1| + |T_2| = N$, $|T_1|, |T_2| \ge N / 3$, and additionally $s \in T_1, t \in T_2$. This affects the size of $S_0$. Now we have 
    \begin{align*}
        \nonumber |S_0| &\ge \sum_{k=\lceil N/3 \rceil}^{\lfloor N/2-1 \rfloor} \binom{N-2}{k-1}k^{k-2} (N-k)^{N-k-2}. 
    \end{align*}
    Because $\binom{N-2}{k-1} = \binom{N}{k} \cdot \frac{k(N-k)}{N(N-1)} = \Omega(\binom{N}{k})$ for $k$ in the specified range, we still obtain the same asymptotic lower bound for $|S_0|$. Hence, the  rest of the proof follows verbatim.
\end{proofsketch}

\section{Acknowledgment}
X.~L.~is supported by NSF Award CCF-2127575.
Y.~X.~is supported by NSF HDR TRIPODS Phase II grant 2217058 (EnCORE Institute).

\bibliographystyle{plainnat}
\bibliography{ref}

\end{document}